\documentclass[11pt]{article}

\usepackage[utf8]{inputenc} 

\usepackage{amsfonts} 
\usepackage{amsmath} 
\usepackage{amssymb}
\usepackage{amsthm} 
\usepackage{subfigure,epsfig,color,soul}
\usepackage{graphicx}
\usepackage{bm}
\usepackage{enumitem}

\usepackage[normalem]{ulem}
\usepackage{color}
\usepackage{tikz}
\usetikzlibrary{arrows}

\newtheorem{theorem}{Theorem}
\newtheorem{lemma}{Lemma}
\newtheorem{proposition}{Proposition}

\newtheorem{definition}{Definition}

\newcommand{\R}{{\mathbb R}}

\newcommand{\X}{\mathcal{X}}
\newcommand{\x}{\textsc{\bf x}}
\newcommand{\y}{\textsc{\bf y}}
\newcommand{\p}{\textbf{p}}
\newcommand{\bl}{\bm{\lambda}}
\newcommand{\bp}{\bm{\pi}}
\newcommand{\conv}{\text{conv}}
\newcommand{\aff}{\text{aff}}

\newcommand{\modif}[1]{{\color{black} #1}}
\newcommand{\edit}[1]{{\color{black} #1}}

\newcounter{claims}
\newcommand{\clitem}[1]{\refstepcounter{claims}\item{\bfseries Claim~\theclaims: \mathversion{bold}#1\mathversion{normal}}}
\newenvironment{claims}
   {\setcounter{claims}{0}
   \description
   }
   {\let\item\origitem\enddescription}

\title{On the combinatorics of the 2-class classification problem}

\author{\sc{Ricardo C. Corrêa} \\
{\small Departamento de Ciência da Computação}\\
{\small Instituto Multidisciplinar}\\
{\small Universidade Federal Rural do Rio de Janeiro}\\
{\small Brazil}
\and \sc{Diego Delle Donne}, \sc{Javier Marenco} \\
{\small Instituto de Ciencias}\\
{\small Universidad Nacional de General Sarmiento} \\
{\small Argentina}}

\begin{document}

\maketitle

\begin{abstract}
A set of points $\X = \X_B \cup \X_R \subseteq \R^d$ is \emph{linearly separable} if the convex hulls of $\X_B$ and $\X_R$ are disjoint, hence there
exists a hyperplane separating $\X_B$ from $\X_R$. Such a hyperplane provides a method for classifying new points, according to which side of the
hyperplane the new points lie. When such a linear separation is not possible, it may still be possible to partition $\X_B$ and $\X_R$ into
prespecified numbers of \emph{groups}, in such a way that every group from $\X_B$ is linearly separable from every group from $\X_R$. We may also
discard some points as \emph{outliers}, and seek to minimize the number of outliers necessary to find such a partition. Based on these ideas,
Bertsimas and Shioda proposed the \emph{classification and regression by integer optimization} (CRIO) method in 2007. In this work we explore the
integer programming aspects of the classification part of CRIO, in particular theoretical properties of the associated formulation. We are able to
find facet-inducing inequalities coming from the stable set polytope, hence showing that this classification problem has exploitable combinatorial properties.

\end{abstract}

\section{Introduction}
\label{sec.introduction}

The data classification problem is a widely-studied topic within the machine learning and the data mining communities. Briefly speaking, this problem aims at finding a partition of an euclidean space that represents the underlying pattern of a set of points. Many computational methods for tackling this problem exist, including decision trees \cite{breiman1984}, linear and quadratic programming~\cite{CarrizosaMorales2012,FreedGlover2007}, and support vector machines \cite{vapnik1999}, among others. Only recently has the use of discrete optimization, and in particular integer programming, been proposed within these settings~\cite{AmaldiConiglioTaccari2016,bertsimas2007,Ryoo2006}.

In this work we are interested in the method called \emph{classification and regression via integer optimization} (CRIO), proposed by Bertsimas and
Shioda in \cite{bertsimas2007}. This pioneering work presents a remarkable application of integer programming to this field, by proposing
to classify points by hyperplanes separating pairs of groups of points. Given two sets of points in $\R^d$, it is not always possible to find a hyperplane separating them, as Figure~\ref{fig.examplea} shows. However, assuming that the underlying pattern can be expressed in terms of convex sets, it may be possible to subdivide each set into \emph{groups}, and then find hyperplanes separating each pair of groups coming from different sets, as in Figure~\ref{fig.exampleb}. In \cite{bertsimas2007} a method is proposed to find such \emph{linearly separable} groups, which includes the solution of a mixed integer program as a key step. This model also identifies a number of points that deviate from the underlying pattern and need to be disregarded in order to enable the desired classification. Further developments~\cite{AmaldiConiglioTaccari2016,Ryoo2006} and computational experiments performed with instances from the literature show that this approach is promising as an effective alternative in practice.

\begin{figure}
\centering
\begin{subfigure}[Two sets of points in $\R^2$.]{ \label{fig.examplea}
\includegraphics[height=4cm]{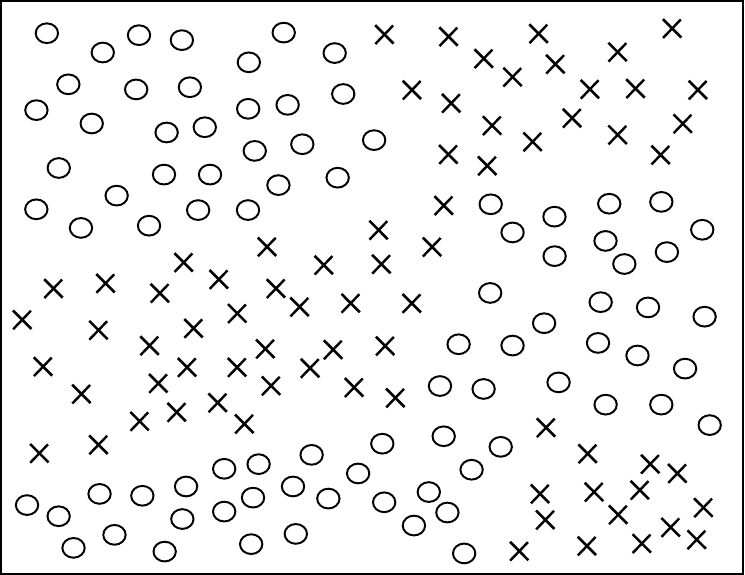}
}\end{subfigure}
\begin{subfigure}[A partition of sets into groups.]{ \label{fig.exampleb}
\includegraphics[height=4cm]{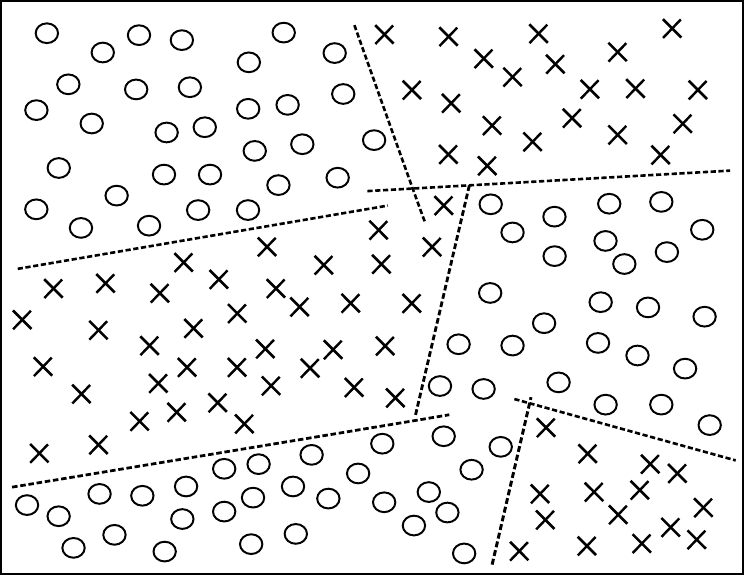}
}\end{subfigure}
\caption{Circles represent points in one set, crosses
represent points in the other set, and the dashed lines represent the hyperplanes separating each group.}
\label{fig.example}
\end{figure}

There exist general algorithms for solving integer programs, and very strong implementations of these algorithms are available. However, since the
general integer programming problem is NP-hard, the running times of these approaches may be intractably large, depending on the instance size and the
model structure. Due to this fact, it is usual to study polyhedra associated with specific integer programming formulations, in order to find strong
valid inequalities that may be helpful within these algorithms. Such polyhedral explorations may also reveal interesting structures and relations
among different problems. In this work we are interested in such issues concerning a 0-1 integer program derived from CRIO in order to partition each set into groups, find the separating hyperplanes among them, and identify a minimum number of disregarded points. The 0-1 integer program studied results from a projection of a mixed integer one. The main goal that motivated this work is the question of whether there exist standard combinatorial
structures inherent to the CRIO method.

The remainder of this work is organized as follows. Section~\ref{sec.problem} introduces the problem in detail and sets out the notation used
throughout the paper. Section~\ref{sec.polyhedral} contains our initial polyhedral study, including several families of facet-inducing inequalities.
A first affirmative answer to the existence of combinatorial structures within the associated polyhedra, in the form of facets coming from the
stable set polytope, is the subject of Section~\ref{sec.ograph}. Finally, Section~\ref{sec.conclusions} closes the paper with conclusions and
lines for future research.

\newcommand{\I}{\mathcal{I}}
\newcommand{\QQ}{Q}
\newcommand{\PP}[1][{\I}]{P_{#1}}

\section{The 2-class problem} \label{sec.problem}

Let $\X = \{ \x_1, \dots, \x_m \} \subset \R^d$ be a set of {\em samples} (also referred to as {\em points}) and consider a partition of
the index set $[m] := \{1, \dots, m\}$ into two \emph{classes} $B$ and $R$, defining subsets $\X_B = \{ \x_i : i \in B \}$ and $\X_R = \{
\x_j : j \in R \}$. We start with some preliminary definitions before stating the mixed integer formulation.

\subsection{Linear separability}

The set $\X$ is {\em linearly separable} if and only if $\conv(\X_B) \cap \conv(\X_R) = \emptyset$, where $\conv(X)$ denotes the convex hull of $X$.
It is worth noting that being linearly separable is equivalent to be partitionable into two convex sets respecting $B$ and $R$. Hence, the
linear separability of $\X$ is characterized by the fact that the set of all $\lambda_i \geq 0$, $i \in [m]$, such that
\begin{align}
\sum\limits_{i \in B} \lambda_i \x_i & = \sum\limits_{j \in R} \lambda_j \x_j, \nonumber \\
\sum\limits_{i \in B} \lambda_i & = \sum\limits_{j \in R} \lambda_j = 1 \nonumber
\end{align}
must be empty. Applying Farkas' Lemma, we get that this characterization is
equivalent to state that there exist $\p \in \R^d$ and $q,r \in
\R$ such that
\begin{align}
r - q & < 0, \nonumber \\
\p\x_i + q & \leq 0 && \mbox{for } i \in B, \nonumber \\
\p\x_j + r & \geq 0 && \mbox{for } j \in R. \nonumber
\end{align}
Adding $q$ in both sides of the last inequality and defining $\delta = q-r$,
we get
\begin{align}
\p\x_i + q & \leq 0 && \mbox{for } i \in B, \nonumber \\
\p\x_j + q & \geq \delta && \mbox{for } j \in R. \nonumber
\end{align}
If these conditions hold for some $\p \in \R^d$ and $q,\delta \in \R$, then set $\delta' = \delta/2$, $q'
= q-\delta'$ and divide $\p$, $q'$, and $\delta'$ by $\delta'$ to conclude
that $\X$ is linearly separable if and only if the set $Q$ of hyperplanes $(\p, q)$
such that
\begin{align}
\p\x_i + q & \leq -1 && \mbox{for } i \in B, \nonumber \\
\p\x_j + q & \geq 1 && \mbox{for } j \in R \nonumber
\end{align}
is not empty. The example in Figure~\ref{fig.example} clearly does not satisfy this property and, then, is not linearly separable.

\subsection{Piecewise linear separability}

In several situations of interest, the underlying pattern of the points in $\X$ cannot be expressed by two convex sets only. In order to
cope with such a scenario, let $L_B$ and $L_R$, $L_B \cap L_R = \emptyset$, be two sets of \emph{group indices} specified for $B$ and $R$,
respectively.
An assignment of points in $\X_B$ to indices in $L_B$ defines groups of points in $\X_B$. \emph{Group} $k
\in L_B$ is the subset of $\X_B$ assigned to index $k$. Groups of points in $\X_R$ are
defined similarly. A \emph{piecewise linear separation} of $\X$ is an
assignment of points in $\X_B$ to indices in $L_B$ and of points in
$\X_R$ to indices in $L_R$ such that groups $k$ and $\ell$ are linearly
separable, for all $k \in L_B$ and $\ell \in L_R$. A hyperplane
$(\p_{k\ell}, q_{k\ell})$ separating groups $k$ and $\ell$ is
such that
\begin{align}
\p_{k\ell}\x_i + q_{k\ell} & \leq -1 && \mbox{for } i \in B \mbox{ such that } \x_i \text{ is assigned to } k, \nonumber \\
\p_{k\ell}\x_j + q_{k\ell} & \geq 1 && \mbox{for } j \in R \mbox{ such that } \x_j \text{ is assigned to } \ell. \nonumber
\end{align}

\subsection{Mixed integer programming formulation} \label{sec.formulation}

We now state a mixed integer programming formulation inspired by the one that constitutes the key step in the CRIO method~\cite{bertsimas2007}. The input of the problem is formed by the set of points $\X$, the partition of point indices $B$ and $R$, and the sets $L_B$ and $L_R$ of group indices. The objective is to find an assignment of points to groups that induces a piecewise linear separation of $\X$.

For every group $k\in L_B$ and every group $\ell\in L_R$, the formulation contains the variables $\p_{k\ell} \in \R^d$ and $q_{k\ell} \in \R$, in such
a way that $\p_{k\ell} \x = q_{k\ell}$ is the hyperplane separating the groups $k$ and $\ell$. For $i\in B$ and $k\in L_B$, the binary variable
$z_{ik}$ represents whether $\x_i$ is assigned to group $k$ or not. For $j\in R$ and $\ell\in L_R$, the binary variable $z_{j\ell}$ represents whether
$\x_j$ is assigned to group $\ell$ or not. In this setting, we can provide the formulation corresponding to the maximization of
\begin{equation}
\sum_{i\in B} \sum_{k\in L_B} z_{ik} + \sum_{j\in R} \sum_{\ell\in L_R} z_{j\ell}
\label{eq.objective}
\end{equation}
subject to
\begin{align}
\p_{k\ell} \x_i + q_{k\ell} &\leq M - (M+1) z_{ik}
			&& \forall i \in B,\forall  k\in L_B, \forall \ell\in L_R,
			\label{eq.formulation.blue}\\
\p_{k\ell} \x_j + q_{k\ell} &\geq - M + (M+1) z_{j\ell}
			&& \forall j \in R,\forall  \ell\in L_R, \forall k\in L_B,
			\label{eq.formulation.red}\\
\sum_{k \in L_B} z_{ik} &\leq 1
			&& \forall i \in B, \label{eq.assign.blue}\\
\sum_{\ell \in L_R} z_{j\ell} &\leq 1
			&& \forall j \in R, \label{eq.assign.red}\\
(\p_{k\ell}, q_{k\ell}) &\in \R^{d+1}
			&& \forall k \in L_B, \forall \ell\in L_R, \label{eq.reals}\\
z_{ik} &\in \{0,1\}
			&& \forall (i,k) \in (B\times L_B) \cup (R\times L_R), \label{eq.binaries}
\end{align}
where $M$ is a big positive number. Note that the feasibility of the model does not depend on the actual value of $M$, namely if $M$ is small then some solutions are lost but the problem remains feasible.

\begin{definition}
Given an instance $\I = (\X, B, R, L_B, L_R)$ of the problem, we call $\PP$ the convex hull of the points $(\p, q, z)$ satisfying
\eqref{eq.formulation.blue}-\eqref{eq.binaries}.
\end{definition}

In addition, we can introduce a binary variable $o_i$ for each $i \in [m]$ to specify whether sample $\x_i$ is an {\em outlier} or not, i.e., if
it is not assigned to any group. With these settings, the objective function can be written as $$
\text{minimize } \sum_{i \in [m]} o_{i}
$$
and constraints (\ref{eq.assign.blue}) and (\ref{eq.assign.red}) should be replaced by the following constraints
\begin{align}
\sum_{k \in L_B} z_{ik} &= 1 - o_i
& \forall i \in B, \label{eq.outliers.blue}\\
\sum_{l \in L_R} z_{j\ell} &= 1 - o_i
& \forall j \in R. \label{eq.outliers.red}
\end{align}

\subsection{Integer programming formulation}

We now discuss an integer programming formulation resulting from a projection of \eqref{eq.formulation.blue}-\eqref{eq.binaries} onto the space of the $z$-variables. For this purpose, let $Q_{k\ell}$ denote the set of points $(\p_{k\ell}, q_{k\ell}, z_{Bk}, z_{R\ell})$ satisfying
\eqref{eq.formulation.blue}-\eqref{eq.formulation.red} and \eqref{eq.reals}-\eqref{eq.binaries} for $\ell\in L_R$ and $k\in L_B$, where $z_{Bk}$ and
$z_{R\ell}$ are vectors constituted by the variables $z_{ik}$, for all $i \in B$, and $z_{j\ell}$, for all $j \in R$, respectively. Rewrite
\eqref{eq.formulation.blue}-\eqref{eq.formulation.red} for $k$ and $\ell$ as
\[
\left[\begin{array}{cc}
\X_B^\top & \mathbf{1} \\
-\X_R^\top & -\mathbf{1}
\end{array}\right]
\left[\begin{array}{c}
\p_{k\ell} \\
q_{k\ell}
\end{array} \right] +
(M+1)\cdot\left[\begin{array}{c}
z_{Bk} \\
z_{R\ell}
\end{array}\right] \leq M\cdot\mathbf{1},
\]
and denote by $Q$ the combination of $Q_{k\ell}$, for all $\ell\in L_R$ and $k\in L_B$. The projection of $Q_{k\ell}$ onto $\{0, 1\}^{B \cup R}$ is defined as
\[
Proj_z(Q_{k\ell}) = \{ (z_{Bk},z_{R\ell}) \in \{0, 1\}^{B \cup R}: \exists (\p_{k\ell},q_{k\ell},z_{Bk},z_{R\ell}) \in Q_{k\ell} \}
\]
The combination of $Proj_z(Q_{k\ell})$, for all $\ell\in L_R$ and $k\in L_B$, gives $Proj_z(Q)$, which in turn gives the set of all possibly
intersecting group assignments. By Theorem 1.1 of~\cite{balas2005}, $Proj_z(Q_{k\ell})$ is given by the group assignments $z$ satisfying
\[
(M+1)\left(\sum_{i \in B} \upsilon_{ik} z_{ik} + \sum_{j \in R} \upsilon_{j\ell} z_{j\ell} \right) \leq M\left(\sum_{i \in B} \upsilon_{ik}
+ \sum_{j \in R} \upsilon_{j\ell} \right)
\]
for all $(\upsilon_{Bk},\upsilon_{R\ell}) \geq \mathbf{0}$ such that
\begin{align}
\sum_{i \in B} \upsilon_{ik} \x_i & = \sum_{j \in R} \upsilon_{j\ell} \x_j \label{eq.combination.ux} \\
\sum_{i \in B} \upsilon_{ik} & = \sum_{j \in R} \upsilon_{j\ell}. \label{eq.combination.u}
\end{align}
It is worth noting that the convex hull of groups $k$ and $\ell$ intersect in a group assignment $z$ if and only if there exists
$(\upsilon_{Bk},\upsilon_{R\ell})$ with $(z_{ik}=0 \Rightarrow \upsilon_{ik}=0)$, $(z_{j\ell}=0 \Rightarrow \upsilon_{j\ell}=0)$, and $\sum_{i \in B} \upsilon_{ik} = 1$ such
that~\eqref{eq.combination.ux}-\eqref{eq.combination.u} hold. Hence,
\begin{equation}
(M+1)\left(\sum_{i \in B} \upsilon_{ik} z_{ik} + \sum_{j \in R} \upsilon_{j\ell} z_{j\ell} \right) \leq 2M
\label{eq.upsilon.z}
\end{equation}
prevents such a group assignment to be chosen. The integer programming formulation consists in maximizing~\eqref{eq.objective} over all
binary points $z$ of type~\eqref{eq.binaries} satisfying~\eqref{eq.assign.blue}-\eqref{eq.assign.red} and \eqref{eq.upsilon.z} for the extreme rays
of the set defined by~\eqref{eq.combination.ux}-\eqref{eq.combination.u}.

A final remark is in order with respect to this integer programming formulation. The projection of $\PP$ onto $[0, 1]^{(B \times L_B) \cup (R
\times L_R)}$, $Proj_z(\PP)$, can be seen as the convex hull of the group assignments in $Proj_z(Q)$
satisfying~\eqref{eq.assign.blue}-\eqref{eq.assign.red}. In addition, the valid inequalities discussed in the next sections involve $z$-variables
only. Consequently, Corollary 2.2 of~\cite{balas2005} can then be applied to conclude that the facetness results of those sections are valid for
$Proj_z(\PP)$ as well.

\newcommand{\zeros}{\hbox{\textbf{0}}}
\newcommand{\nil}{\hbox{\textbf{b}}}
\newcommand{\base}{\hbox{\textbf{b}}}
\newcommand{\unitp}{e^{\p}}
\newcommand{\unitq}{e^{q}}
\newcommand{\unitz}{e^{z}}

\section{Polyhedral study} \label{sec.polyhedral}

	In this section we are interested in facets of $\PP$, with a particular interest in combinatorial structures originating facet-inducing inequalities.
	For $k\in L_B$, $\ell\in L_R$, and $a\in[d]$, we denote by $\unitp_{k\ell {a}}$ the unit vector associated with the ${a}$-th
	coordinate of the variable vector $\p_{k\ell}$. Correspondingly, for $k\in L_B$ and $\ell\in L_R$, we denote by $\unitq_{k\ell}$ the unit vector associated with the variable $q_{k\ell}$. For $i\in B$ and $k\in L_B$, let $\unitz_{ik}$ be the unit vector associated with the variable $z_{ik}$. Finally, for $j\in R$ and $\ell\in L_R$, we denote by $\unitz_{j\ell}$ the unit vector associated with the variable $z_{j\ell}$.

	\begin{proposition}
	\label{prop.dimension}
		$\PP$ is full-dimensional.
	\end{proposition}

	\begin{proof}
	In order to prove this proposition, we construct the following affinely independent feasible solutions.
	\begin{enumerate}
	\item Let $\nil = (\zeros, \zeros, \zeros)$ be the solution having all variables set to null values. Since $z_{ik} = 0$ in $\nil$ for every $i\in [m]$, then no point is assigned to any group, and all constraints are satisfied, hence $\nil$ is feasible.
	\item For any $k\in L_B$, $\ell\in L_R$, and ${a}\in[d]$, consider the solution $\nil + \varepsilon \unitp_{k\ell {a}}$, with $0 <
	\varepsilon$ \edit{and $\varepsilon < M / |\x_{ia}|$ if $\x_{ia} \ne 0$, where $\x_{ia}$ is the $a$-th coordinate of $\x_i$}.
	This solution is feasible since all points are outliers due to \edit{the fact that all $z$ variables are null}, and it is affinely
	independent w.r.t.~the previous solutions, which have $\p_{k\ell a} = 0$.
	\item Similarly, for any $k\in L_B$ and $\ell\in L_R$, the solution $\nil + \varepsilon \unitq_{k\ell}$ is feasible and affinely independent w.r.t.~the previously-constructed solutions.
	\item For any \edit{$i'\in B$} and $k'\in L_B$, construct the solution $\nil + \edit{\unitz_{i'k'}} - \sum_{\ell\in L_R} \unitq_{k'\ell}$. Constraint \eqref{eq.formulation.blue} for $k=k'$, \edit{$i=i'$,} and $\ell\in L_R$ takes the form \modif{$-1\le -1$} (since $q_{k'\ell} = -1$ in this solution), hence it is satisfied. The remaining constraints are trivially satisfied. Furthermore, this solution is affinely independent w.r.t.~the previous solutions, which have $z_{ik'} = 0$.
	\item Similarly, for any \edit{$j'\in R$ and $\ell'\in L_R$, the solution $\nil + \unitz_{j'\ell'} - \sum_{k\in L_B} \unitq_{k\ell'}$} is feasible and affinely independent w.r.t.~the previous ones.
	\end{enumerate}
	The existence of these solutions shows that $\PP$ is full-dimensional.
 	\end{proof}

	The solutions constructed within the proof of Proposition~\ref{prop.dimension} allow to show the following facetness results in a quite straightforward way, hence the proof of the following proposition is omitted.

	\begin{proposition}
	\begin{itemize}
	\item[(i)] The model constraints \eqref{eq.assign.blue} and \eqref{eq.assign.red} are facet-inducing.
	\item[(ii)] The bound $z_{ik} \ge 0$ is facet-inducing, for every $i\in B$ and $k\in L_B$.
	\item[(iii)] The bound $z_{j\ell} \ge 0$ is facet-inducing, for every $j\in R$ and $\ell\in L_R$.
	\end{itemize}
	\end{proposition}

	\subsection{Convex-inclusion inequalities}

	We now explore families of valid inequalities for $\PP$, and study their facetness properties. We first present a familiy of valid inequalities involving a point $\x_j\in\X_R$ and a set of points in $\X_B$ whose convex hull contains $\x_j$. In this setting, we may consider the valid inequality given by the following proposition. We adopt the notation $\x_S := \{\x_i: i\in S\}$, for any $S \subseteq [m]$.

	\begin{proposition}
		Let $k \in L_B$, $j \in R$ and ${S} = \{i_1, \ldots, {i_s}\} \subseteq B$ such that $\x_{j} \in \conv(\x_{S})$.
		The \emph{convex-inclusion inequality}
		\begin{equation}
			\sum_{i\in {S}} z_{ik} + \sum_{\ell \in L_R} z_{j\ell} \ \leq\ {s}
			\label{in.convex.inclusion}
		\end{equation}
		is valid for $\PP$.
	\end{proposition}

	\begin{proof}
	Let $(\p,q,z)\in\PP$ be a feasible solution. If $\x_j$ is an outlier in this solution ({\em i.e.}, $z_{j\ell} = 0$ for every $\ell\in L_R$), then
	\eqref{in.convex.inclusion} is trivially satisfied, so assume $z_{j\ell} = 1$ for some $\ell\in L_R$. Since $\x_j$ is contained in the convex hull of
	the points $\x_{S}$, then there is no hyperplane separating $\x_j$ from $\x_{S}$. This implies that a solution having all the points in
	${S}$ assigned to the same group would not be feasible, hence $\sum_{i\in {S}} z_{ik} \le |{S}|-1 = {s}-1$, and
	\eqref{in.convex.inclusion} is satisfied.
	\end{proof}

	Xavier and Camp\^{e}lo \cite{xavier2011} proposed a general facet-generating procedure that takes a valid (facet-inducing) inequality $\pi x \leq \pi_0$ for a polytope $P$ and a valid (facet-inducing) inequality for the face of $P$ defined by $\pi x = \pi_0$, and produces a new valid (facet-inducing) inequality for $P$ by combining them. It is interesting to note that inequalities \eqref{in.convex.inclusion} are obtained with the proposed procedure by using \eqref{eq.assign.red} as $\pi x \leq \pi_0$ along with $\sum_{i\in {S}} z_{ik} \leq {s} {-1}$, which is valid when $\sum\limits_{\ell \in L_R} z_{j\ell} = 1$.

	\edit{If $(\p,q,z)$ is a feasible solution, we say that a constraint is \emph{strictly satisfied} by $(\p,q,z)$ if the latter does not satisfy the
	constraint with equality.}

	\begin{theorem}
	Assume $|L_R| \ge 2$. The inequality \eqref{in.convex.inclusion} defines a facet of $\PP$ if and only if ${S}$ is minimal w.r.t.~the property $\x_{j} \in \conv(\x_{S})$ ({\em i.e.}, $\x_{j} \not \in \conv(\x_{{S}'})$ for every ${S}'\subset {S}$).
	\end{theorem}

	\begin{proof}
	Assume first that ${S}$ is minimal w.r.t.~the property $\x_{j} \in \conv(\x_{S})$. Let $F$ be the face of $\PP$ defined by
	\eqref{in.convex.inclusion}, and let $({\bl},\mu,\gamma)$ and $\lambda_0$ be such that ${\bl} \p + \mu q + \gamma z = \lambda_0$ for
	every $(\p,q,z) \in F$. We shall verify that $({\bl},\mu,\gamma)$ is a multiple of the coefficient vector of \eqref{in.convex.inclusion}, thus
	showing that $F$ is a facet of $\PP$. To this end, let $\base = (\zeros, q, z)$ be the solution obtained by setting $z_{ik} = 1$ for every $i\in
	{S}$, and all the remaining $z$-variables to 0 ({\em i.e.}, all the points in $\X \backslash \x_{S}$ are outliers). We set $q_{k\ell} = -2$ for all
	$\ell\in L_R$, and the remaining $q$-variables are set to 0. Note that constraints \eqref{eq.formulation.blue} for $i\in {S}$ and corresponding to
	the group{s $k$ and $\ell$} are strictly satisfied \edit{({\em i.e.}, without equality)}. This point is feasible and satisfies
	\eqref{in.convex.inclusion} with equality.

\begin{claims}
\clitem{$\lambda = 0$.}\label{it.claim1.1}
	For $k'\in L_B$, $\ell\in L_R$, and ${a}\in[d]$, consider the solution $\base + \varepsilon \unitp_{k'\ell {a}}$, which is
	feasible if $\varepsilon$ is small enough. Indeed, constraints \eqref{eq.formulation.blue} for $k = k'$ and $i\in {S}$ are satisfied since they
	are strictly satisfied by $\base$. In this solution only one group is nonempty, and the solution satisfies \eqref{in.convex.inclusion} with equality. Together with $\base$, the existence of this solution implies $\lambda_{k'\ell {a}} = 0$. $\Diamond$

\clitem{$\mu = 0$.}\label{it.claim1.2}
	Similarly, for $k'\in L_B$ and $\ell\in L_R$, the solution $\base + \varepsilon \unitq_{k'\ell}$ is feasible (for a small
	enough $\varepsilon$) and also satisfies \eqref{in.convex.inclusion} with equality. Again, the combination of this solution with $\base$ implies
	$\mu_{k'\ell} = 0$. $\Diamond$

\clitem{$\gamma_{ik} = \gamma_{j\ell}$ for $i\in {S}$ and $\ell\in L_R$.}\label{it.claim1.3}
	Let ${S}' := {S}\backslash\{i\}$ and consider the solution $\base_{i\ell}$ constructed by setting $z_{i'k} = 1$ for $i'\in {S}'$,
	$z_{j\ell} = 1$, and the remaining $z$-variables to 0 (including $z_{ik} = 0$). Since $\x_{j} \not \in \conv(\x_{{S}'})$, then there exists a
	hyperplane ${\bp \y} = \pi_0$ separating $\x_j$ from $\x_{{S}'}$. Let $\p_{k\ell} = {\bp}$ and $q_{k\ell} = - \pi_0$, \edit{set $q_{k\ell'}
	= -2$ for all $\ell' \ne \ell$,} and set the remaining $\p$- and $q$-variables to 0. The solution $\base_{i\ell}$ thus constructed is
	feasible and satisfies \eqref{in.convex.inclusion} with
	equality, hence $({\bl},\mu,\gamma) \base = ({\bl},\mu,\gamma) \base_{i\ell}$. Since ${\bl} = \bm{0}$ and $\mu = 0$, we conclude
	that $\gamma_{ik} = \gamma_{j\ell}$. $\Diamond$

\clitem{$\gamma_{j'\ell} = 0$ for $j'\ne j$ and $\ell\in L_R$.}\label{it.claim1.4}
	If $\x_{j'}$ is not contained in $\conv(\x_{S})$, then there exists a hyperplane separating $\x_{j'}$ from $\x_{S}$. Otherwise, if
	$\x_{j'}$ is contained in $\conv(\x_{S})$, since the points in $\x_{S}$ are affinely independent (due to the minimality of ${S}$),
	then there exists ${i} \in {S}$ such that $\x_{j'} \notin \conv(\x_{S} \backslash\{\x_{i}\})$. In both cases,
	there exists some ${i} \in {S}$ such that we can find a hyperplane separating $\x_{j'}$ from $\conv(\x_{S}
	\backslash\{\x_{i}\})$.
	Construct a solution $\base_{j'\ell}$ by setting $z_{i'k} = 1$ for all $i'\in {S}\backslash \{{i}\}$, $z_{j'\ell'} = 1$, and
	$z_{j\ell} = 1$ for any $\ell \in L_R \setminus \{\ell'\}$ (recall that $|L_R| \ge 2$), and the remaining $z$-variables to 0.
	Finally, set $(\p_{k\ell}, q_{k\ell})$ equal to the hyperplane separating $\x_{S} \backslash\{\x_{i}\}$ from $\x_j$, and $(\p_{k\ell'},
	q_{k\ell'})$ equal to the hyperplane separating $\x_{S} \backslash\{\x_{i}\}$ from $x_{j'}$. This solution is feasible and satisfies
	\eqref{in.convex.inclusion} with equality. By resorting to Claim~\ref{it.claim1.1}, Claim~\ref{it.claim1.2}, and Claim~\ref{it.claim1.3}, the
	existence of $\base$ and $\base_{j'\ell'}$ shows that $\gamma_{j'\ell'} = 0$. $\Diamond$

\clitem{$\gamma_{ik'} = 0$ for $i\in {S}$ and $k'\ne k$.}\label{it.claim1.5}
	Take any $\ell \in L_R$ and consider the solution $\base_{k'}$ obtained from $\base - \unitz_{ik} + \unitz_{ik'} + \unitz_{j\ell}$ by setting
	$(\p_{k\ell}, q_{k\ell})$ equal to the hyperplane separating $\x_{S} \backslash\{\x_i\}$ from $\x_j$, \edit{and $(\p_{k'\ell}, q_{k'\ell})$
	equal to the hyperplane separating $\x_i$ from $\x_j$}. This solution is feasible and satisfies \eqref{in.convex.inclusion} with equality.
	By resorting to the previous claims, the existence of $\base_{k'}$ and $\base$ shows that $\gamma_{ik'} = 0$. $\Diamond$

\clitem{$\gamma_{ik'} = 0$ for $i\not\in {S}$ and $k'\in L_B$.}\label{it.claim1.6}
	Similarly to the proof of Claim~\ref{it.claim1.5}, the solution $\base_{ik'} = \base + \unitz_{ik'}$ is feasible and satisfies
	\eqref{in.convex.inclusion} with equality, so $\gamma_{ik'} = 0$. $\Diamond$
\end{claims}
	By combining these claims, we conclude that $({\bl},\mu,\gamma)$ is a multiple of the coefficient vector of \eqref{in.convex.inclusion}, hence
	$F$ is a facet of $\PP$.

	For the converse direction, suppose that $\x_{j} \in \conv(\x_{S} \backslash \{\x_{i}\})$ for some $i\in {S}$. This implies that any
	solution having $z_{ik} = 0$ cannot satisfy \eqref{in.convex.inclusion} with equality, since such a solution must have $z_{i'k} = 1$ for every $i'\in
	{S}\backslash\{i\}$ and $z_{j\ell} = 1$ for some $\ell \in L_R$, in order to attain equality. However, since $\x_{j} \in \conv(\x_{S} \backslash \{\x_{i}\})$ then no hyperplane separating $\x_j$ from $\conv(\x_{S} \backslash \{\x_{i}\})$ exists, hence such a solution cannot be feasible. This implies that every solution in the face of $\PP$ induced by \eqref{in.convex.inclusion} satisfies $z_{ik} = 1$ and, since $\PP$ is full-dimensional, such a face is not a facet of $\PP$.
	\end{proof}

	The symmetrical inequalities considering a point $\x_i \in \X_B$ and a set of points in $\X_R$ whose convex hull includes $\x_i$ have the same properties as \eqref{in.convex.inclusion}.

\subsection{Obstacle inequalities}

	Given two distinct points ${\y_1}, {\y_2} \in \R^d$, a set ${Y} \subseteq \R^d$ is an \textit{obstacle}
	between ${\y_1}$ and ${\y_2}$ if $\conv({Y}) \cap \conv(\{{\y_1}, {\y_2}\}) \neq \emptyset$ (see
	Figure~\ref{fig.obstacle}). We say that ${Y}$ is a \emph{trivial obstacle} if ${\y_1} \in \conv({Y})$ or
	${\y_2} \in \conv({Y})$, and that ${Y}$ is a \emph{minimal obstacle} if $\conv({Y}\setminus\{{\y}_i\}) \cap
	\conv(\{{\y_1}, {\y_2}\}) = \emptyset$, for every ${\y}_i \in {Y}$. \edit{We denote by $\aff(Y)$ the affine space
	generated by the points in $Y$.}

	The presence of an obstacle \edit{that is a subset of $\X_B$} between two points of $\X_R$ implies that $\X_B$ and $\X_R$ are not linearly separable. However, it is
	interesting to remark that the converse is not true. There may not exist an obstacle between two points of $\X_R$ even when $\X_B$ and $\X_R$ are not linearly separable; although in such a case an obstacle will exist between some pair of points of $\conv(\X_R)$ (recall that $\X_B$ and $\X_R$ are not polytopes but finite sets of points). An example of this is given by the sets
	\begin{align*}
	\X_B &= \{ (1,1,0,0), (-2,1,0,0), (1,-2,0,0) \}\\
	\X_R &= \{ (0,0,1,1), (0,0,-2,1), (0,0,1,-2) \}
	\end{align*}
	for which $\conv(\X_B) \cap \conv(\X_R) = \{(0,0,0,0)\}$.


	\pgfmathsetmacro{\sincoordl}{2*sin(72)}
	\pgfmathsetmacro{\coscoordl}{2*cos(72)}
	\pgfmathsetmacro{\sincoords}{2*sin(36)}
	\pgfmathsetmacro{\coscoords}{2*cos(36)}

	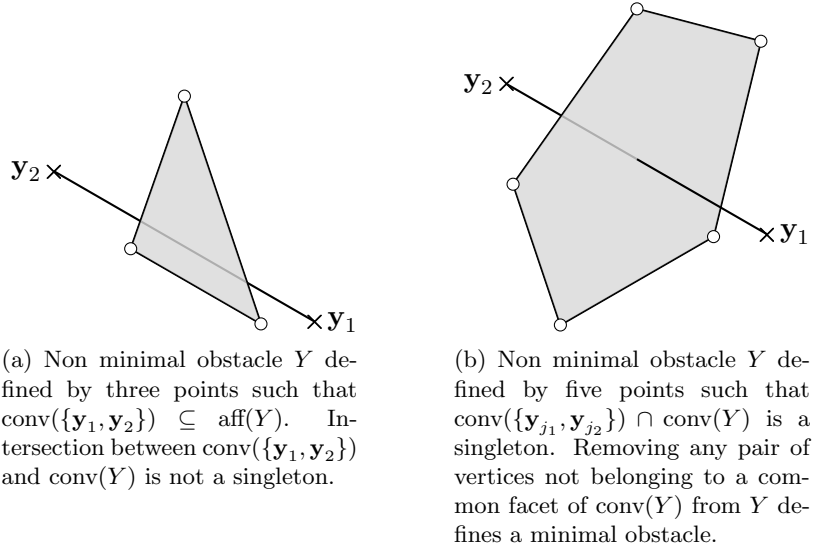
\begin{figure}
	\centering
	\begin{subfigure}[Non minimal obstacle $Y$ defined by three points such that $\conv(\{ \y_1, \y_2 \})
	\subseteq \aff(Y)$. Intersection between $\conv(\{ \y_1, \y_2 \})$ and $\conv(Y)$
	is not a singleton.]{ \label{fig.minimalo} \begin{tikzpicture}[x={(0.866cm,-0.5cm)}, y={(0.866cm,0.5cm)}, z={(0cm,1cm)}, scale=1.0,
    >=stealth, %
    inner sep=0pt, outer sep=2pt,%
    axis/.style={thick,color=black, draw opacity=1},
    obstacle/.style={fill=black!15!white, fill opacity=0.8},
]
    \draw[axis,-] (-2,0,0) -- (2, 0, 0);

	\coordinate (o1) at (-2, \sincoords,  -\coscoords);
	\coordinate (o3) at (0, 0, 2);
	\coordinate (o5) at (0, \sincoords, -\coscoords);
    \draw[axis] (o1)  -- (o3) -- (o5) -- cycle;
    \draw[obstacle] (o1)  -- (o3) -- (o5) -- cycle;

    \foreach \i in {1,3,5}
	  	\draw[fill=white,fill opacity=1] (o\i) circle (0.8mm);
    \draw (2, 0, 0) node {$\times$} node {$\times$} node {$\times$} node[right,xshift=0.1cm] {$\y_1$};
    \draw (-2, 0, 0) node {$\times$} node {$\times$} node {$\times$} node[left,xshift=-0.1cm] {$\y_2$};
	\end{tikzpicture}}
	\end{subfigure}\quad\quad\quad
	\begin{subfigure}[Non minimal obstacle $Y$ defined by five points such that $\conv(\{\y_{j_1}, \y_{j_2}\}) \cap \conv(Y)$ is a singleton. Removing
	any pair of vertices not belonging to a common facet of $\conv(Y)$ from $Y$ defines a minimal obstacle.]{ \label{fig.nonminimalo} \begin{tikzpicture}[x={(0.866cm,-0.5cm)}, y={(0.866cm,0.5cm)},
	z={(0cm,1cm)}, scale=1.0,
    >=stealth, %
    inner sep=0pt, outer sep=2pt,%
    axis/.style={thick,color=black, draw opacity=1},
    obstacle/.style={fill=black!15!white, fill opacity=0.8},
]
    \draw[axis,-] (-2,0,0) -- (0, 0, 0);

	\coordinate (o1) at (0, -\sincoords,  -\coscoords);
	\coordinate (o2) at (0, -\sincoordl, \coscoordl);
	\coordinate (o3) at (0, 0, 2);
	\coordinate (o4) at (0, \sincoordl, \coscoordl);
	\coordinate (o5) at (0, \sincoords, -\coscoords);
    \draw[axis] (o1)  -- (o2) -- (o3) -- (o4) -- (o5) -- cycle;
    \draw[obstacle] (o1)  -- (o2) -- (o3) -- (o4) -- (o5) -- cycle;

    \foreach \i in {1,...,5}
	  	\draw[fill=white,fill opacity=1] (o\i) circle (0.8mm);
    \draw (2, 0, 0) node {$\times$} node {$\times$} node {$\times$} node[right,xshift=0.1cm] {$\y_1$};
    \draw (-2, 0, 0) node {$\times$} node {$\times$} node {$\times$} node[left,xshift=-0.1cm] {$\y_2$};
    \draw[axis,-] (0,0,0) -- (2, 0, 0);
	\end{tikzpicture}}
	\end{subfigure}
		\caption{\label{fig.obstacle}Two obstacles (and the convex hull of its points) between the points $\y_1$ and $\y_2$ in $\mathbb{R}^3$.}
	\end{figure}

	\begin{proposition}
	Let $S = \{i_1, \ldots, i_s\} \subseteq B$ be such that $\x_S$ is an
	obstacle between two points $\x_{j_1}, \x_{j_2} \in \X_R$, $j_1\ne j_2$. For $k \in L_B$ and $\ell \in L_R$, the \emph{obstacle inequality}
	\begin{equation}
	z_{j_1\ell} + z_{j_2\ell} + \sum_{i\in {S}} z_{ik} \leq {s} + 1
	\label{in.obstacle}
	\end{equation}
	is valid for $\PP$.
	\end{proposition}

	\begin{proof}
	Let $(\p,q,z)\in\PP$ be a feasible solution. Since the left-hand side of \eqref{in.obstacle} contains ${s}+2$ binary variables, we
	need only consider the case $z_{j_1\ell} = z_{j_2\ell} = 1$ and $z_{ik} = 1$ for every $i\in {S}$. These variable values imply that $\x_{j_1}$
	and $\x_{j_2}$ are assigned to the group $\ell$, whereas all the points in $\x_S$ are assigned to the group $k$. This is not possible in a feasible solution, since $\conv(\x_S) \cap \conv(\{\x_{j_1}, \x_{j_2}\}) \neq \emptyset$ implies that there is no hyperplane separating $\{\x_{j_1}, \x_{j_2}\}$ from $\x_S$. Hence, $(\p,q,z)$ satisfies \eqref{in.obstacle}. Since $(\p,q,z)$ is an arbitrary solution, then \eqref{in.obstacle} is valid for $\PP$.
	\end{proof}

	\edit{We now explore the facetness of the obstacle inequalities. To this end, we first state the following preliminary lemmas.}

	\edit{
	\begin{lemma} \label{lem.affine}
	Let $S \subseteq B$ be such that $\x_S$ is an obstacle between $\x_{j_1}, \x_{j_2} \in \X_R$, $j_1\ne j_2$. If $\conv(\{ \x_{j_1}, \x_{j_2} \})
	\subseteq \aff(\x_S)$, then $\x_S$ is trivial or non minimal.
	\end{lemma}

	\begin{proof}
	Assume that $\x_S$ is nontrivial. Let $F$ be a facet of $\conv(\x_S)$ containing a point of $\conv(\{\x_{j_1}, \x_{j_2}\})$.
	Such a facet exists since otherwise $\conv(\{\x_{j_1}, \x_{j_2}\}) \cap \conv(\x_S)$ is a
	polyhedron defined by a system of linear equations (this implies that $\conv(\{\x_{j_1}, \x_{j_2}\}) \cap \conv(\x_S)$ is a singleton like in
	Figure~\ref{fig.nonminimalo}), contradicting either the hypothesis that $\x_S$ is nontrivial or the hypothesis asserting that $\conv(\{ \x_{j_1},
	\x_{j_2} \}) \subseteq \aff(\x_S)$. Therefore, the set of vertices of $F$ is a proper subset of $\x_S$ and forms an obstacle between $\x_{j_1}$ and $\x_{j_2}$.
	\end{proof}
	}

	\begin{lemma} \label{lem.obstacle}
	Let $S \subseteq B$ be such that $\x_S$ is a
	nontrivial \edit{minimal} obstacle between $\x_{j_1}, \x_{j_2} \in \X_R$, $j_1\ne j_2$. Let $i\in B {\setminus S}$, and define
	$S' = S \cup \{i\}$. Then, $\x_{j_1} \not\in\conv({\x_{S'}})$ or $\x_{j_2} \not\in\conv({\x_{S'}})$.
	\end{lemma}

	\begin{proof}
	\edit{If $\x_i \in \conv(\x_S)$, then the lemma trivially holds. Thus, assume that $\x_i \not\in \conv(\x_S)$.
	\modif{If $\x_{j_2} \in \aff(\x_S)$, then $\x_{j_1}$ could be written as an affine combination of $\x_{j_2}$ and any point $\x$ in
	$\conv(\{\x_{j_1}, \x_{j_2}\}) \cap \conv({\x_S})$ (such an $\x \ne \x_{j_2}$ exists since $\x_S$ is a nontrivial obstacle between
	$\x_{j_1}$ and $\x_{j_2}$). The situation where $\x_{j_1} \in \aff(\x_S)$ is analogous with the roles of $\x_{j_1}$ and
	$\x_{j_2}$ interchanged. Therefore, $\x_{j_1} \in \aff(\x_S)$ or $\x_{j_2} \in \aff(\x_S)$ implies $\conv(\{ \x_{j_1}, \x_{j_2}
	\}) \subseteq \aff(\x_S)$, which contradicts Lemma~\ref{lem.affine}.}
	
	So assume $\x_{j_1},\x_{j_2} \not\in \aff(\x_S)$. It follows from $\conv(\{\x_{j_1}, \x_{j_2}\}) \cap \aff({\x_S}) \ne \emptyset$ that there exists a hyperplane $(\p, q)$ separating $\x_{j_1}$ from $\x_{j_2}$ such that $\p \x_{j_1} > q$,
	$\p \x_{j_2}
	< q$, and $\p \x = q$ for all $\x \in \aff(\x_S)$. Additionally, $\p \y > q$ or $\p \y < q$ is violated by $\x_i$ and by all points in
	$\conv({\x_{S}})$ simultaneously. In the former case, $\x_{j_1} \not\in\conv({\x_{S'}})$ since $\p \y \leq q$ holds for all points in
	$\conv({\x_{S'}})$ and, in the latter case, $\x_{j_2} \not\in\conv({\x_{S'}})$.}
	\end{proof}

	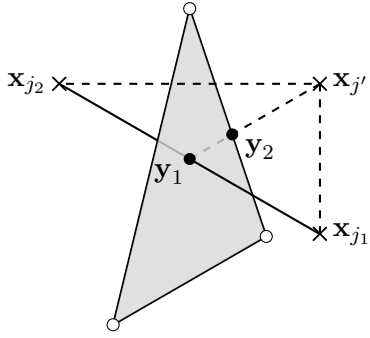
\begin{figure}
	\centering
	\begin{tikzpicture}[x={(0.866cm,-0.5cm)}, y={(0.866cm,0.5cm)},
	z={(0cm,1cm)}, scale=1.0,
    >=stealth, %
    inner sep=0pt, outer sep=2pt,%
    axis/.style={thick,color=black, draw opacity=1},
    obstacle/.style={fill=black!15!white, fill opacity=0.8},
]
	\coordinate (o6) at (0, 0, 0);
	\coordinate (o7) at (0, 2, 0);
	\coordinate (o8) at (2, 0, 0);
	\coordinate (o9) at (-2, 0, 0);

    \draw[axis,-] (-2,0,0) -- (0, 0, 0);
    \draw[axis,dashed] (o6) -- (o7);

    \draw[axis] (o1) -- (o3) -- (o5) -- cycle;
    \draw[obstacle] (o1)  -- (o3) -- (o5) -- cycle;

    \foreach \i in {1,3,5}
	  	\draw[fill=white,fill opacity=1] (o\i) circle (0.8mm);
    \draw (o8) node {$\times$} node {$\times$} node {$\times$} node[right,xshift=0.1cm] {$\x_{j_1}$};
    \draw (o9) node {$\times$} node {$\times$} node {$\times$} node[left,xshift=-0.1cm] {$\x_{j_2}$};
    \fill[black] (o6) circle (0.8mm) node[left,yshift=-0.2cm] {$\y_1$};
    \draw (o7) node {$\times$} node {$\times$} node {$\times$} node[right,xshift=0.1cm] {$\x_{j'}$};
    \fill[black] (0,\sincoords/1.8,0) circle (0.8mm) node[right,xshift=0.1cm,yshift=-0.2cm] {$\y_2$};

    \draw[axis,-] (0,0,0) -- (2, 0, 0);
    \draw[axis,dashed] (o8) -- (o7) -- (o9);
	\end{tikzpicture}
		\caption{An obstacle (and the convex hull of its points) between the points $\x_{j_1}$ and $\x_{j_2}$ in $\mathbb{R}^3$.}
		\label{fig.intersing}
	\end{figure}

	\begin{lemma} \label{lem.separable}
	\edit{Let $S\subseteq B$ be such that $\x_S$ is a nontrivial minimal obstacle between $\x_{j_1}, \x_{j_2} \in \X_R$, $j_1\ne j_2$,
	and let $j'\in R$, $j'\ne j_1, j_2$. Then, there exists $i\in S$ such that $\conv(\{\x_{j_1}, \x_{j_2}, \x_{j'}\})$ and $\conv(\x_{S'})$
	are linearly separable, where $S' = S\backslash\{i\}$.}
	\end{lemma}

	\begin{proof}
	\edit{Lemma~\ref{lem.affine} and the fact that $\x_S$ is a nontrivial minimal obstacle yield that the intersection between
	$\conv(\{\x_{j_1}, \x_{j_2}\})$ and $\conv(\x_S)$ is a singleton, say $\{\y_1\}$ (this is so because a second point $\y'$ in
	$\conv(\{\x_{j_1}, \x_{j_2}\}) \cap \conv(\x_S)$ would result in $\conv(\{ \x_{j_1}, \x_{j_2} \})
	\subseteq \aff(\x_S)$). Resorting to similar arguments, we conclude that $\conv(\{\x_{j_1}, \x_{j'}\})
	\cap \aff(\x_S)$ is either empty or a singleton. The same applies to $\conv(\{\x_{j_2}, \x_{j'}\}) \cap \aff(\x_S)$. Additionally, if \modif{$\x_{j'}
	\not\in \aff(\x_S)$} then either $\conv(\{\x_{j_1}, \x_{j'}\}) \cap \aff(\x_S)$ or $\conv(\{\x_{j_2}, \x_{j'}\}) \cap \aff(\x_S)$ is empty. It
	follows that the intersection between $\conv(\{\x_{j_1}, \x_{j'}\}) \cup \conv(\{\x_{j_2}, \x_{j'}\})$ and $\aff(\x_S)$ is also a singleton, which means that
	$\conv(\{\x_{j_1}, \x_{j_2}, \x_{j'}\}) \cap \conv(\x_S) = \conv(\{\y_1,\y_2\})$, for some $\y_2\in\conv(\x_S)$, as illustrated in
	Figure~\ref{fig.intersing}.
	If $\y_2 \ne \y_1$ then there exists $i\in S$ such that $\y_2 \not\in \conv(\x_{S'})$, where $S' = S \backslash \{i\}$, because the points in $\x_S$
	are affinely independent (due to the minimality of $\x_S$). Furthermore, the minimality of $\x_S$ implies that $\y_1 \not\in \conv(\x_{S'})$.
	Therefore, $\conv(\{\y_1,\y_2\}) \cap \conv(\x_{S'}) = \emptyset$, and the lemma follows.}
	\end{proof}

	We say that a hyperplane ${\bp} {\y} \le \pi_0$ \emph{strictly separates} the sets ${Y}\subseteq\R^d$ and ${Y'}\subseteq\R^d$
	if ${\bp} {\y} < \pi_0-1$ for every ${\y}\in {Y}$ and ${\bp} {\y} > \pi_0+1$ for every ${\y}\in
	{Y'}$. If ${Y}=\{{\y_1}\}$ is a singleton, we also say that ${\bp} {\y}\le \pi_0$ \emph{strictly separates}
	${\y_1}$ from ${Y'}$ if ${\bp} {\y}\le \pi_0$ strictly separates ${Y}$ and ${Y'}$. Note that a strictly-separating
	hyperplane between two sets of points exists whenever those sets \edit{are linearly separable}.

	\begin{theorem}
	The inequality \eqref{in.obstacle} defines a facet of $\PP$ if and only if ${\x_S}$ is a nontrivial minimal obstacle between $\x_{j_1}$ and
	$\x_{j_2}$.
	\end{theorem}

	\begin{proof}
	Assume first that ${\x_S}$ is nontrivial and minimal. Let $F$ be the face of $\PP$ defined by \eqref{in.obstacle}, and let
	$({\bl},\mu,\gamma)$ and $\lambda_0$ such that ${\bl} \p + \mu q + \gamma z = \lambda_0$ for every $(\p,q,z) \in F$.
	Define $\base = (\p, q, z)$ to be the solution obtained by setting $z_{ik} = 1$ for every $i\in {S}$, $z_{j_1\ell} = 1$, and all the remaining
	$z$-variables to 0. Also, set $\p_{k\ell}$ and $q_{k\ell}$ equal to a hyperplane strictly separating $\x_{S}$ from $z_{j_1\ell}$ (which exists
	since $\x_{j_1} \not\in \conv(\x_{S})$), $q_{k\ell'} = -2$ for $\ell' \ne \ell$, and the rest of the ${\p}$- and $q$-variables equal to
	0. We clearly have $\base\in F$.

\begin{claims}
\clitem{$\lambda = 0$.}\label{it.claim2.1}
	For $k'\in L_B$, $\ell'\in L_R$, and ${a}\in[d]$, consider the solution $\base_{k'\ell' {a}} := \base + \varepsilon \unitp_{k'\ell'
	{a}}$, for a small enough $\varepsilon$, which is feasible since only two groups are nonempty, and are strictly separated in $\base$ (so a
	minor change in one coefficient keeps the separation). The existence of the solutions $\base$ and $\base_{k'\ell' {a}}$ in $F$ implies
	$\lambda_{k'\ell' {a}} = 0$. $\Diamond$

\clitem{$\mu = 0$.}\label{it.claim2.2}
	Similarly, for $k'\in L_B$ and $\ell'\in L_R$, the solution $\base_{k'\ell'} := \base + \varepsilon \unitq_{k'\ell'}$ is
	feasible and also satisfies \eqref{in.obstacle} with equality, \edit{for $\varepsilon$ small enough}.
	Again, the existence of $\base$ and $\base_{k'\ell'}$ implies $\mu_{k'\ell'} = 0$. $\Diamond$

\clitem{$\gamma_{ik} = \gamma_{j_1\ell} = \gamma_{j_2\ell}$ for $i\in {S}$.}
	Let ${S}' := {S}\backslash\{i\}$ and consider the solution $\base_i$ constructed by setting $z_{i'k} =
	1$ for $i'\in {S}'$, $z_{j_1\ell} = z_{j_2\ell} =1$, and the remaining $z$-variables to 0. Since \edit{${\x_{S'}}$ is minimal, then}
	${\x_{S'}}$ is not an obstacle between
	$\x_{j_1}$ and $\x_{j_2}$, so there exists a hyperplane ${\bp \y} = \pi_0$ separating $\{\x_{j_1}, \x_{j_2} \}$ from $\x_{{S}'}$. Let
	$\p_{k\ell} = {\bp}$, $q_{k\ell} = -\pi_0$, \edit{$q_{k\ell'} = -2$} for $\ell' \ne \ell$, $q_{k'\ell} = 2$ for $k' \ne k$, and set the remaining
	$\p$- and $q$-variables to 0. Claim~\ref{it.claim2.1} and Claim~\ref{it.claim2.2} ensure that ${\bl} = {\bm{0}}$ and $\mu = 0$, hence
	the existence of $\base_i$ and $\base$ in $F$ implies $\gamma_{ik} = \gamma_{j_2\ell}$. A symmetric argument shows that $\gamma_{ik} = \gamma_{j_1\ell}$. $\Diamond$

\clitem{$\gamma_{j_1\ell'} = \gamma_{j_2\ell'} = 0$ for $\ell'\in L_R$, $\ell'\ne\ell$.}
	Construct a solution $\base_{\ell'}$ by setting $z_{j_2\ell'} = 1$, and setting $(\p_{k\ell'}, q_{k\ell'})$ as a hyperplane strictly separating
	$\conv(\x_{S})$ from $\x_{j_2}$, \edit{which exists since $\x_S$ is a nontrivial obstacle}. This solution is feasible and, together with $\base$,
	Claim~\ref{it.claim2.1}, and Claim~\ref{it.claim2.2}, shows that $\gamma_{j_2\ell'} = 0$.
	A symmetric argument shows that \edit{$\gamma_{j_1\ell'} = 0$.} $\Diamond$

\clitem{$\gamma_{ik} = 0$ for $i\not\in {S}$.}\label{it.claim2.5}
	Let ${S}' := {S}\cup\{i\}$. Since ${\x_S}$ is a nontrivial obstacle between $\x_{j_1}$ and $\x_{j_2}$, then
	Lemma~\ref{lem.obstacle} implies that $\x_{j_1} \not\in \conv(\x_{{S}'})$ or $\x_{j_2} \not\in \conv(\x_{{S}'})$. If $\x_{j_1} \not\in
	\conv(\x_{S'})$, then the solution $\base_{i} := \base + \unitz_{ik}$ is feasible since there is a hyperplane separating
	$\x_{j_1}$ from $\x_{S'}$, and the existence of $\base$ and $\base_{i}$ \edit{(together with Claim~\ref{it.claim2.1} and
	Claim~\ref{it.claim2.2})} implies the claim. If $\x_{j_2} \not\in\conv(\x_{S'})$ then a symmetrical argument
	also settles the claim. $\Diamond$

\clitem{$\gamma_{j'\ell'} = 0$ for $j'\ne j_1, j_2$ and $\ell'\in L_R$.}
	\edit{Lemma~\ref{lem.separable} ensures that there exists $i\in S$ such that $\x_{S'}$ and $\{\x_{j_1}, \edit{\x_{j_2}}, \x_{j'}\}$
	are linearly separable, where $S' = S \backslash \{i\}$. Construct a solution $\base_{j'\ell'}$ by setting $z_{i'k} = 1$ for all
	$i'\in {S'}$, setting $z_{j_1\ell} = \edit{z_{j_2\ell}} = z_{j'\ell'} = 1$, and setting the remaining $z$-variables to 0.
	\begin{itemize}
	\item If $\ell\ne\ell'$, then set $(\p_{k\ell}, q_{k\ell})$ equal to the hyperplane separating $\x_{S'}$ from $\x_p$, and $(\p_{k\ell'},
	q_{k\ell'})$ equal to the hyperplane separating $\x_{S'}$ from $x_{j'}$.
	\item If $\ell = \ell'$, then set $(\p_{k\ell}, q_{k\ell})$ equal to the hyperplane separating $\x_{S'}$ from $\{\x_{j_1}, \x_{j_2},
	\x_{j'}\}$, whose existence is guaranteed by Lemma~\ref{lem.separable}.
	\end{itemize}
	This solution is feasible and satisfies \eqref{in.obstacle} with equality. By resorting to Claim~\ref{it.claim2.1} and
	Claim~\ref{it.claim2.2}, the existence of $\base$ and $\base_{j'\ell'}$ shows that $\gamma_{j'\ell'} = 0$.} $\Diamond$

\clitem{$\gamma_{ik'} = 0$ for $i\in {S}$ and $k'\ne k$.}
	Define $S' = S \backslash \{i\}$. By the minimality of ${\x_S}$, the solution $\base_{ik}$ obtained by setting $z_{i'k} = 1$ for $i'\in
	{S'}$, $z_{j_1\ell} = z_{j_2\ell} = 1$, and $z_{ik'} = 1$ is feasible, since the sets ${S'}$ and $\{\x_{j_1}, \x_{j_2}\}$ are linearly
	separable, and so are the sets $\{\x_i\}$ and $\{\x_{j_1}, \x_{j_2}\}$. Hence we can set the variables $(\p_{k\ell}, q_{k\ell})$ and $(\p_{k'\ell},
	q_{k'\ell})$ to the corresponding separating hyperplanes. Since ${\bl} = {\bm{0}}$ and $\mu = 0$, the existence of $\base$ and
	$\base_{ik'}$ implies that $\gamma_{ik'} = 0$. $\Diamond$

\clitem{$\gamma_{ik'} = 0$ for $i\not\in {S}$ and $k'\ne k$.}
	Similarly to the proof of Claim~\ref{it.claim2.5}, the solution $\base_{ik'} = \base + \unitz_{ik'}$ is feasible and satisfies \eqref{in.obstacle}
	with equality, so $\gamma_{ik'} = 0$. $\Diamond$
\end{claims}
	By combining these claims, we conclude that $({\bl},\mu,\gamma)$ is a multiple of the coefficient vector of \eqref{in.obstacle}, hence $F$ is a
	facet of $\PP$.

	For the converse direction, if ${\x_S}$ is trivial with $\x_{j_1} \in \conv({\x_S})$, then \eqref{in.obstacle} is the sum of $z_{j_2\ell}
	\leq 1$ and the inequality \eqref{in.convex.inclusion} associated with $\x_{j_1}$ and ${S}$ (analogously if $\x_{j_2} \in
	\conv({\x_S})$).
	On the other hand, if ${S}$ is not minimal, say ${S' := S}\backslash\{i'\}$ for some $i'\in {S}$ is still an obstacle between
	$\x_{j_1}$ and $\x_{j_2}$, then any solution having $z_{i'k} = 0$ cannot satisfy \eqref{in.obstacle} with equality, since such a solution must have
	$z_{ik} = 1$ for every $i\in {S'}$ and $z_{j_1\ell} = z_{j_2\ell} = 1$ in order to attain equality. However, in this setting no hyperplane
	separating $\{\x_{j_1}, \x_{j_2}\}$ from $\conv(\x_{S'})$ exists, hence such a solution is not feasible. We thus conclude that every solution in the face of $\PP$ induced by \eqref{in.obstacle} satisfies $z_{i'k} = 1$, hence this face is not a facet of $\PP$.
	\end{proof}

	The symmetrical inequalities considering two points $\x_{i_1}, \x_{i_2} \in \X_B$ and a set of points in $\X_R$ whose convex hull includes $\conv(\{\x_{i_1}, \x_{i_2}\})$ have the same properties as \eqref{in.obstacle}.

	Inequalities \eqref{in.convex.inclusion} and \eqref{in.obstacle} are based on similar ideas, and we can consider a natural generalization of these
	inequalities in the following way. If $S\subseteq B$ and $T\subseteq R$ are two sets of points with $\conv(\x_S) \cap \conv(\x_T) \ne \emptyset$, then the inequality
	\begin{displaymath}
	\sum_{i\in S} z_{ik} + \sum_{j\in T} z_{j\ell} \ \le\ |S| + |T| - 1
	\end{displaymath}
	is valid for $\PP$ for every $k\in L_B$ and $\ell\in L_R$. However, this inequality may not be not facet-inducing when $\conv(\x_S) \cap \conv(\x_T)$ has nonzero dimension.

\section{The obstacle graph} \label{sec.ograph}

\edit{Let $V\subseteq R$ be the index set of a subset of points in $\X_R$. The graph $G=(V,E)$ is an \emph{obstacle graph} for $V$ if, for each
$jj'\in E$ there exists \modif{a subset $S_{jj'} \subseteq B$ such that $\x_{S_{jj'}}$ is an obstacle between $\x_j$ and $\x_{j'}$}. A \emph{stable set} in the graph $G = (V,E)$ is a set $I\subseteq V$ of vertices such that $jj'\not\in E$ for every $j,j'\in I$. The maximum cardinality of a stable set in $G$ is the \emph{stability number} of $G$, and is denoted by $\alpha(G)$.}

\begin{proposition} \label{propo.ograph}
\edit{Let $V \subseteq R$ and let $G=(V,E)$ be an obstacle graph for $V$. Let $k:E\to L_B$ be a function associating a group $k_e \in L_B$ to each edge $e\in E$. For $\ell \in L_R$, the \emph{obstacle rank inequality}}
\begin{equation}
\sum_{j \in V} z_{j\ell} \ \le \ \alpha(G) + \sum_{e\in E} \Big( |S_e| - \sum_{i\in S_e} z_{ik_e} \Big), \label{in.rank}
\end{equation}
\edit{is valid for $\PP$.}
\end{proposition}

\begin{proof}
\edit{Let $(\p,q,z) \in \PP$ be a feasible solution and $G' = (V,E')$ be the subgraph of $G$ defined by $E' = \{e\in E: \sum_{i \in S_e} z_{ik_e} =
|S_e|\}$. Note that $jj' \in E'$ implies that $\x_j$ and $\x_{j'}$ cannot belong to the same group in $L_R$ (although the converse implication may not
hold, {\em e.g.}, if the obstacle between them is not minimal), hence the variables $\{z_{j\ell}\}_{j\in V}$ induce a stable set in $G'$. For any
$e\in E$, we have $\alpha(G\backslash e) \le \alpha(G) + 1$, hence $\alpha(G') \le \alpha(G) + |E\backslash E'|$. This implies}
\begin{eqnarray}
\sum_{j \in V} z_{j\ell} & \le & \alpha(G') \ \le \ \alpha(G) + |E\backslash E'| \nonumber \\
                         & \le & \alpha(G) + \sum_{e\in E\backslash E'} \Big( |S_e| - \sum_{i\in S_e} z_{ik_e} \Big) \nonumber \\
                         & \modif{=} & \alpha(G) + \sum_{e\in E} \Big( |S_e| - \sum_{i\in S_e} z_{ik_e} \Big). \nonumber
\end{eqnarray}
\edit{Since $(\p,q,z)$ is an arbitrary solution, then \eqref{in.rank} is valid for $\PP$.}
\end{proof}

\edit{The {\em incidence vector} of a stable set $I$ is ${x}^I \in \{0,1\}^{|V|}$ defined by ${x}^I_i = 1$ if and only if $i\in I$. The \emph{stable
set polytope} of a graph $G$ is $STAB(G) = \conv\{ {x}^I: I$ is a stable set of $G\}$. For an obstacle graph $G=(V,E)$ for $V$, define $\mathcal{S} =
\cup_{e\in E} S_e$. We say that the obstacle graph $G=(V,E)$ is \emph{disjoint} if $S_e \cap S_{e'} = \emptyset$ for every $e,e'\in E$, $e\ne e'$. The
obstacle graph is \emph{minimal} if $S_{jj'}$ is a minimal obstacle between $\x_j$ and $\x_{j'}$, for every $jj'\in E$. An edge $e\in E$ is
\emph{critical} if $\alpha(G\backslash e) = \alpha(G) + 1$, and we say that $G$ is \emph{critical} if every edge in $E$ is
critical}. \modif{It is known that a critical and connected graph is also a rank-facet inducing graph, {\em i.e.,} $\sum_{j\in V} x_j \le \alpha(G)$
induces a facet of $STAB(G)$~\cite{Chvatal75}}. \modif{A simple example of a disjoint, minimal, and critical obstacle graph is depicted in
Figure~\ref{fig.ograph}.}

	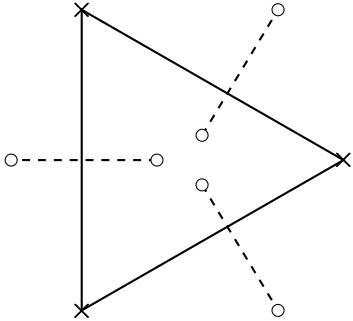
\begin{figure}
	\centering
	\begin{tikzpicture}[scale=2.0,
    >=stealth, %
    inner sep=0pt, outer sep=2pt,%
    axis/.style={thick,color=black, draw opacity=1},
    obstacle/.style={fill=black!15!white, fill opacity=0.8},
]
	\coordinate (v1) at (0, 0, 0);
	\coordinate (v2) at (0, 2, 0);
	\coordinate (v3) at (1.74, 1, 0);
	\coordinate (v4-1) at (0.8, 1.165, 0);
	\coordinate (v4-2) at (0.8, 0.835, 0);
	\coordinate (v4-3) at (0.5, 1, 0);
	\coordinate (v5) at (1.305, 2, 0);
	\coordinate (v6) at (1.305, 0, 0);
	\coordinate (v7) at (-0.469, 1, 0);


    \draw[axis] (v1) -- (v2) -- (v3) -- cycle;
    \foreach \i/\j in {1/5,2/6,3/7} {
	    \draw[axis,dashed] (v4-\i) -- (v\j);
	    \draw[fill=white,fill opacity=1] (v4-\i) circle (0.4mm);
	    \draw[fill=white,fill opacity=1] (v\j) circle (0.4mm);
	  	\draw[draw opacity=1] (v\i) node {$\times$} node {$\times$} node {$\times$};
	}
	\end{tikzpicture}
		\caption{The critical obstacle graph defined by three disjoint and minimal obstacles (indicated with dashed lines).}
		\label{fig.ograph}
	\end{figure}

\begin{theorem} \label{thm.rank}
\edit{Let $V \subseteq R$, $|V| > 1$, and consider a disjoint, minimal, critical, and connected obstacle graph $G=(V,E)$ for $V$. Let $k:E\to L_B$
\modif{and $\mathcal S_t = \{ i \in \mathcal S: i\in S_e$ for some $e\in E$ with $k_e = t \}$ for every $t \in L_B$}.
If \begin{enumerate}[label=({\it \roman*})]
\item \label{hyp.separable} $\x_I$ and \modif{$\x_{\mathcal S_t}$} are linearly separable, for every maximum-size stable set $I$ of $G$ and every
\modif{$t\in L_B$},
\item \label{hyp.separable2} for every $j \in R \setminus V$, there exists a maximum-size stable set $I$ of $G$ such that \modif{$\x_{I\cup\{j\}}$} is \modif{linearly} separable from \modif{$\x_{\mathcal S_t}$, for every $t\in L_B$}, and
\item \label{hyp.separable3} for every $i\in B\backslash\mathcal{S}$ and every \modif{$t\in L_B$}, there exists a maximum-size stable set $I$ of $G$ such that \modif{$\x_I$} is \modif{linearly} separable from $\x_{S'}$ where $S' := \modif{S_t}\cup\{i\}$,
\end{enumerate}
then \eqref{in.rank} induces a facet of $\PP$.}
\end{theorem}

\begin{proof}
Let $F$ be the face of $\PP$ defined by \eqref{in.rank}, and suppose that ${\bl} \p + \mu q + \gamma z = \lambda_0$ for every solution $(\p,q,z) \in F$. We shall prove that $({\bl}, \mu, \gamma)$ is a multiple of the coefficient vector of \eqref{in.rank}, thus showing that $F$	is indeed a facet of $\PP$.

If $I\subseteq V$ is a stable set, we define $\base^I$ to be the solution given by $z_{j\ell} = 1$ for every $j\in I$, \edit{$z_{ik_e} = 1$ for every
$e\in E$ and every $i \in S_e$, and by setting the remaining $z$-variables to 0. In addition, the hypothesis~\ref{hyp.separable} ensures that there
exists a hyperplane strictly separating $\{\x_j\}_{j\in I}$ from \modif{$\{\x_{i}\}_{i\in \mathcal S_t}$ for every $t\in L_B$, and we set
$(\p_{t\ell}, q_{t\ell})$} to be such a hyperplane}. The solution thus constructed is feasible and satisfies \eqref{in.rank} with equality.

\begin{claims}
\clitem{$\lambda = 0$ and $\mu = 0$.} \label{it.claim3.1}
Let $I\subseteq V$ be a maximum-size stable set of \modif{$G$}. For $k\in L_B$, $\ell'\in L_R$, and ${a}\in[d]$, consider the solution $\base^I_{k\ell' {a}} := \base^I + \varepsilon \unitp_{k\ell' {a}}$, which is feasible since either $k$ and $\ell'$ are \modif{linearly} separated, or else one of these groups is empty. The existence of these solutions implies $\lambda_{k\ell' {a}} = 0$. A similar argument shows $\mu_{k\ell'} = 0$. $\Diamond$

\clitem{$\gamma_{j\ell} = \gamma_{j'\ell}$ for every $j,j'\in V$.} \label{it.claim3.2}
\modif{Since $|V|>1$ and $G$ is connected, let $j$ and $j'$ be such that $e=jj' \in E$ ({\em i.e.}, we assume first that $j$ and $j'$ are adjacent in $G$). Since $e$ is critical, then there exists a stable set $I'$ of $G\backslash e$ with $|I'| = \alpha(G)+1$ and $j, j' \in I'$. So, consider the points $\base^{I_1}$, $I_1 = I'\setminus \{j\}$, and $\base^{I_2}$, $I_2 = I'\setminus \{j'\}$. The solutions $\base^{I_1}$ and $\base^{I_2}$ satisfy \eqref{in.rank} with equality and only differ in the variables $z_{j\ell}$, $z_{j'\ell}$, $\p$, and $q$. Since $\lambda = 0$ and $\mu = 0$ by Claim~\ref{it.claim3.1}, we get $\gamma_{j\ell} =
\gamma_{j'\ell}$. The claim stems from the fact that $G$ is connected}.
$\Diamond$

\clitem{$\gamma_{ik_e} = \gamma_{j\ell}$ for every ${i}\in S_e$, $e\in E$, and every $j\in V$.} \label{it.claim3.3}
\edit{Assume first that $j$ is an endpoint of $e$. \modif{Consider the stable sets $I'$ of $G\backslash e$ and $I=I_1$ of $G$ constructed in
Claim~\ref{it.claim3.2}.} Let $\base^{I'}_i$ be defined from
$\base^{I'}$ by setting $z_{i'k_e} = 1$ for $i' \in S_e\backslash\{i\}$ and $z_{ik_e} = 0$. Since $S_e$ is a minimal obstacle (by the minimality of $G$), the point $\base^{I'}_i$ is feasible. Also, $\base^{I'}_i$ satisfies \eqref{in.rank} with equality due to the fact that $G$ is disjoint. The existence of the points $\base^I$ and $\base^{I'}_i$, together with Claim~1 and Claim~2, implies that $\gamma_{ik_e} = \gamma_{j\ell}$ \modif{for every ${i}\in S_e$, $e\in E$, and every $j\in V$.} $\Diamond$}

\clitem{$\gamma_{j\ell'} = 0$ for every $j\in V$ and $\ell'\ne\ell$.} \label{it.claim3.4}
\modif{Since $|V| > 1$ and $G$ is connected then $j$ is not isolated. Thus, let $e$ be an edge incident to $j$. Furthermore, let $I'$ and $I = I_1$ be
stable sets as defined in Claim~\ref{it.claim3.2}}. The solution $\base^I_{j\ell'}$ obtained from $\base^I$ by setting $z_{j\ell'} = 1$ is feasible
since the additional group $\ell'$ is \modif{linearly separable from $\mathcal S_t$ for all $t \in L_B$ by the hypothesis~\ref{hyp.separable}. The
variables $\{(\p_{t\ell'}, q_{t\ell'})\}_{t \in L_B}$ are set accordingly in $\base^I_{j\ell'}$}. Since the $z$-variables of \modif{$\base^I$ and
$\base^I_{j\ell'}$} only differ in the value of $z_{j\ell'}$, we conclude that $\gamma_{j\ell'} = 0$. $\Diamond$

\clitem{$\edit{\gamma_{ik}} = 0$ for every $e\in E$, ${i}\in S_e$, and $k\ne k_e$.}
\edit{Let $\base^I$ and \modif{$\base^{I'}_i$} be the solutions constructed in Claim~\ref{it.claim3.3}, and construct \modif{$\base^{I'}_{ik}$} from
\modif{$\base^{I'}_i$} by setting \modif{$z_{ik} = 1$}. This new point is feasible by the hypothesis~\ref{hyp.separable} and, together with $\base^{I'}_i$, implies} $\gamma_{ik} =
0$. $\Diamond$

\clitem{$\gamma_{j\ell} = 0$ for every $j \in R \setminus V$.} \label{it.claim3.6}
The hypothesis\modif{~\ref{hyp.separable2}} ensures that there exists a maximum-size stable set $I$ of $G$ such that \modif{$\x_{I\cup\{j\}}$} is \modif{linearly} separable \modif{from $\x_{S_t}$, for every $t\in L_B$}. Consider the solution $\base^I_j$ defined from $\base^I$ by setting $z_{j\ell} = 1$, and adjusting the variables $\{(\p_{k\ell}, q_{k\ell})\}_{k\in L_B}$ accordingly. This construction, together with $\base^I$ \modif{and
Claim~\ref{it.claim3.1}}, establishes the claim.
$\Diamond$

\clitem{$\gamma_{j\ell'} = 0$ for every $j \in R \setminus V$ and $\ell'\ne\ell$.}
Let $I$ be \modif{the} maximum-size stable set of $G$ \modif{constructed in Claim~\ref{it.claim3.6}}, and  construct the solution $\base^I_{j\ell'}$
from $\base^I$ by setting $z_{j\ell'} = 1$ and adjusting the variables $\{(\p_{k\ell'}, q_{k\ell'})\}_{k\in L_B}$ accordingly. \modif{This new solution is feasible since $\x_{\{j\}} \subseteq \x_{I\cup\{j\}}$ is linearly separable from $x_{S_t}$, for every $t\in L_B$.} The existence of the solutions $\base^I_{j\ell'}$ and $\base^I$	implies, together with Claim~\ref{it.claim3.1}, that $\gamma_{j\ell'} = 0$. $\Diamond$

\clitem{$\gamma_{ik} = 0$ for every $i \in B \setminus \mathcal{S}$ and $k\in L_B$.}
Let $I$ be a maximum-size stable set of $G$ and consider the solution $\base^I_{ik}$ defined by setting $z_{ik} = 1$ and adjusting the variables $\{(\p_{k\ell'}, q_{k\ell'})\}_{\ell'\in L_R}$ accordingly. If $k=k_e$ for some $e\in E$, we take $I$ to be a stable set satisfying the hypothesis~\ref{hyp.separable3}. The existence of this solution, together with $\base^I$, implies $\gamma_{ik} = 0$. $\Diamond$
\end{claims}

By combining these claims, we conclude that $({\bl},\mu,\gamma)$ is a multiple of the coefficient vector of \eqref{in.rank}, hence this inequality
induces a facet of $\PP$.
\end{proof}

\modif{The obstacle graph in Figure~\ref{fig.ograph} satisfies the hypotheses of Theorem~\ref{thm.rank} provided all the remaining points within $\X$ are located outside the convex hull of the depicted points (although the converse is not true in general).} Theorem~\ref{thm.rank} is a step towards an answer of our initial interest in the polytope $\PP$, namely the existence of combinatorial structures within the polytope associated with the classification problem addressed in this work. Indeed, \emph{every} facet of the stable set polytope $STAB(G)$ gives rise to facet-inducing inequalities of $\PP$, provided the technical conditions stated in Theorem~\ref{thm.rank} are satisfied. It is interesting to note that the weigthed rank inequalities discussed in~\cite{CorreaDelleDonneKochMarenco2017} can also be adapted in a similar way.

\section{Concluding remarks}
\label{sec.conclusions}

In this work we have started a polyhedral study of a particular 2-class classification problem asking for a classification of points in $\R^d$ into linearly separable groups, and allowing for outliers not to be classified into any group. Our main interest is exploring the polyhedral structure of an integer programming formulation for this problem, and in this direction we were able to identify facet-inducing inequalities involving the binary variables that specify the points assigned to each group. The valid inequalities presented only involve these binary variables, and hint at combinatorial structures present in the classification problem studied, particularly in $Proj_z(\PP)$.

A main result states that every facet-inducing inequality for the stable set polytope can be translated into facets of $Proj_z(\PP)$. It would
be interesting to explore whether similar ideas, namely considering the groups as colors, can be applied to adapt valid inequalities from vertex
coloring, $k$-partite subgraph, and $k$-partition polytopes to the polytope considered in this work. In this vein, a direction for further
investigation is the polyhedral consequences of combining the structure of $Proj_z(\PP)$ with the orbitope structure in order to eliminate symmetries~\cite{KaibelPfetsch2008}.

An important aspect of our approach is that the integer programming formulation results from a projection of a mixed integer one including the
$\p$- and $q$-continuous variables. The particular structure of feasible solutions allows us to search for procedures that derive new
inequalities from inequalities that are valid for subsets of points. More specifically, let $B'\subseteq B$ and $R'\subseteq R$, and define $Q_{B'R'}
:= \{ \p \in \R^d, q \in \R: \p\x_i + q \leq -1$ for every $\x_i\in B'$ and $\p\x_j + q \geq 1$ for every $\x_j\in R' \}$. By definition, $Q_{B'R'}$ is a polytope, and let $\alpha^T (\p, q) \leq \lambda_0$ be a valid inequality for
	$Q_{B'R'}$. Hence, for any $k\in L_B$ and $\ell\in L_R$, the inequality
	\begin{equation}
	\alpha^T (\p_{k\ell}, q_{k\ell}) \ \le \ \lambda_0
	\label{eq.procedure}
	\end{equation}
is satisfied by the solutions having $z_{ik} = 1$ for every $i\in B'$ and $z_{j\ell} = 1$ for every $j\in R'$. We can now lift these inequalities in order to obtain an inequality valid for $\PP$. A straightforward (although not optimal) lifting is given by
	\begin{displaymath}
	\alpha^T (\p_{k\ell}, q_{k\ell}) \ \le \ \lambda_0 + M' \Big(|B'| + |R'| - \sum_{i\in B'} z_{ik} - \sum_{j\in R'} z_{j\ell} \Big).
	\end{displaymath}
	It would be interesting to explore whether there exist valid inequalities involving the $\p$- and $q$-variables not coming from liftings of
	inequalities of the form \eqref{eq.procedure} and their projections onto the $z$-variables space.

\section*{Acknowledgment}

We are very grateful for the invaluable remarks by the anonymous referee, in
special for pointing out some mistakes in the original version of Lemma~\ref{lem.affine} and Proposition~\ref{propo.ograph}, as well as bringing a
Lemma~\ref{lem.affine} proof's sketch to our attention.

\bibliographystyle{plain}
\bibliography{classification}

\end{document}